\documentclass[12pts]{article}

\usepackage[english]{babel}
\usepackage[utf8x]{inputenc}
\usepackage[T1]{fontenc}

\usepackage[a4paper,top=3cm,bottom=2cm,left=3cm,right=3cm,marginparwidth=1.75cm]{geometry}

\usepackage{amsmath}
\usepackage{amsthm}
\usepackage{amssymb}

\newtheorem{theorem}{Theorem}[section]
\newtheorem{corollary}{Corollary}[theorem]
\newtheorem{lemma}[theorem]{Lemma}
\usepackage{graphicx}
\usepackage[colorinlistoftodos]{todonotes}
\usepackage[colorlinks=true, allcolors=blue]{hyperref}
\usepackage{float}
\usepackage{setspace}
\newtheorem*{remark}{Remark}
\restylefloat{table}
\title{Evaluating the Properties of a First Choice Weighted Approval Voting System}
\author{Peter Butler and Jerry Lin}

\begin{document}
\maketitle
\begin{center}
\textbf{Abstract}\\
\end{center}
Plurality and approval voting are two well-known voting systems with different strengths and weaknesses. In this paper we consider a new voting system we call beta(k) which allows voters to select a single first-choice candidate and approve of any other number of candidates, where $k$ denotes the relative weight given to a first choice; this system is essentially a hybrid of plurality and approval. Our primary goal is to characterize the behavior of beta(k) for any value of $k$. Under certain reasonable assumptions, beta(k) can be made to mimic plurality or approval voting in the event of a single winner while potentially breaking ties otherwise. Under the assumption that voters are honest, we show that it is possible to find the values of $k$ for which a given candidate will win the election if the respective approval and plurality votes are known. Finally, we show how some of the commonly used voting system criteria are satisfied by beta(k).

\begin{center}
\textbf{Introduction}\\
\end{center}

Plurality is a popular voting system where each voter votes for exactly one candidate; close to a third of all countries use plurality for government elections (Ace Project, n.d.). One major drawback of plurality is that it tends to force a two-party system as a result of Duverger's Law (Riker 1982, 753). This makes third-party candidates practically noncontenders in plurality systems. Another issue with plurality arises in elections with more than two candidates. In such cases, a candidate who does not have a majority of the votes may still be elected, as may a candidate who is less preferred overall than other candidates (Brams and Fishburn, 2007, 1-3). In particular, plurality can at times elect more extremist candidates in elections when there are more than two candidates since a candidate does not need a simple majority of the votes to win such an election. This is seen as a problem with plurality by critics. Approval voting is an alternative system where voters may approve of any number of candidates. It has been adopted by several professional groups including the American Mathematical Society and the American Statistical Association (Karlin and Peres, 2017, 223). One clear advantage of approval voting is that minority candidates actually stand a chance to win (Brams and Fishburn, 2007, 7). While approval does not suffer from Duverger's Law as does plurality, approval is not a Pareto-efficient system; a candidate who is unanimously preferred to another candidate is not guaranteed to have a higher approval score (we show this in sub-section 4.4). This is mainly due to the fact that, under approval, voters cannot express the degree to which they prefer one candidate over another. Although there is a scarcity of empirical evidence on the efficacy of approval voting, some critics believe that it promotes more centrist candidates than plurality (Brams and Fishburn, 2007, 10; Cox, 1985, 118).
\bigskip

  We have developed a new voting system, called beta(k), that we believe can function as a reasonable hybrid of plurality and approval. Under beta(k), a voter may approve of candidates while also denoting their top preference. In this system, approvals have a weight of $1$ while first-choice votes are weighted by $k \geq 1$. We believe that this system maintains the benefits of approval voting and mitigates some of the potential drawbacks. Introducing the additional weight for first-choice preferences gives voters more options for expressing their levels of preference and could discourage candidates from simply taking moderate stances on all issues. One of our goals is to find how the value of $k$ affects the outcome of a beta(k) election. We make the reasonable assumptions that any voter who would approve of a candidate in an approval election would approve of that candidate in a beta(k) election, and that any voter who would vote for a candidate in a plurality election would denote that candidate as their first-choice in a beta(k) election. We find values of $k$ that guarantee a beta(k) winner to agree with a plurality winner or an approval winner, assuming a single winner exists, while potentially breaking ties otherwise. We also show that by knowing the approval votes and plurality votes corresponding to a set of voters, we can find the subset of candidates who could potentially win in a beta(k) election, and we can find the ranges of $k$ for which each of these candidates will win. Lastly, we compare the three systems by checking their compliance with several common voting criteria.
 \bigskip
 
Much research has been done on the behavior of different voting systems and many criteria have been developed to determine which systems function "better." In general, while the concept of better and worse voting systems depends on criteria being checked, we believe that some of the most desirable properties are non-dictatorship, monotonicity, unanimity, and Pareto-efficiency (all of which we formally define in section 4).  
\bigskip

\section{Summary of relevant research}
Brams and Fishburn have speculated that some elections would have had quite different results had approval been used instead of plurality (Brams and Fishburn, 2007, 59-69). Karlin and Peres have shown previously that approval and plurality are monotonic (Karlin and Peres, 2017, 221-224).  One issue frequently discussed by theorists is the possibility of voting strategically. The Gibbard-Satterthwaite theorem shows that no system (except dictatorship) is strategy-proof when there are more than two candidates, a famous result from the 1970s that is proven elegantly in (Karlin and Peres, 2017, 226-228). There is theoretical and experimental research on voting strategies for plurality and approval. In a plurality system where voters act strategically, both the Condorcet Winner and Condorcet Loser Criteria are violated (Niou, 2001, 225). However, a recent experimental study with students at New York University suggests that plurality is less manipulable and more socially efficient than approval (Bassi, 2015, 77).

\section{Basic Definitions}
\flushleft{\textbf{Definition 2.1}} Let $\textbf{X} = \{C_1, C_2, \ldots, C_c\}$ be a finite, ordered set with cardinality $c \in \mathbb{N}$ such that $c \geq 1$. We refer to the elements of $\textbf{X}$ as \textbf{candidates}. 

\flushleft{\textbf{Definition 2.2}}
We assume that there exists a preference relation for each voter on the set of candidates. We denote this relation by \textbf{$\succ$}, i.e. $A\succ B$ means that the voter prefers A over B. We assume that this relation is \textbf{complete} (for all candidates A and B, either $A\succ B$ or $B\succ A$) and \textbf{transitive} (if $A\succ B$ and $B \succ C$, then $A\succ C$). Furthermore, we denote the $i^{th}$ voter's preference relation by \textbf{$\succ_i$}, and the collection of preference relations for all voters is a \textbf{preference profile} $(\succ_1,\succ_2,\ldots,\succ_n)$. 

\textit{Note: Definition 2.2 is based on definitions in (Karlin and Peres 2017, 218)}.

\flushleft{\textbf{Definition 2.3}}
A \textbf{vote} over \textbf{X} is a \textit{non-zero} real-valued vector in $\mathbb{R}^c$ such that the $j^{th}$ element in the vector corresponds to candidate $C_j$ in $\textbf{X}$.

\flushleft{\textbf{Definition 2.4}}
Let $k\in \mathbb{R}$ such that $k\geq 1$. A \textbf{beta(k) vote} over \textbf{X} is a type of vote in $\mathbb{R}^c$ such that each element in the vector is in $\{0,1,k\}$ and exactly one element is equal to k.

\flushleft{\textbf{Definition 2.5}}
An \textbf{approval vote} is a type of vote in $\mathbb{R}^c$ such that each element in the vector is in $\{0,1\}$ and at least one element is equal to 1.

\flushleft{\textbf{Definition 2.6}}
A \textbf{plurality vote} is a type of vote in $\mathbb{R}^c$ such that each element in the vector is in $\{0,1\}$ and exactly one element is equal to 1.

\flushleft{\textbf{Definition 2.7}}
A \textbf{vote matrix} is a $n \times c$ matrix of real numbers such that every row within the matrix is a vote of the same type. A \textbf{beta(k) matrix} is a vote matrix in which all the votes are beta(k) votes and every beta(k) vote uses the same $k$ value. An \textbf{approval matrix} is a vote matrix in which all the votes are approval votes. A \textbf{plurality matrix} is a vote matrix in which all the votes are plurality votes. 

\flushleft{\textbf{Definition 2.8}}
Let $K(P,A)$ denote a function that takes in a $n \times c$ plurality matrix $P$ and a $n \times c$ approval matrix $A$ and outputs a $n \times c$  beta(k) matrix $B$ such that: 

$$B = K(P,A) = P\cdot(k-1) + A.$$

\flushleft{\textbf{Definition 2.9}}
Let $E(\textbf{X},Z)$ denote a function that takes in a set of candidates $\textbf{X}$ with cardinality $c$ and a $n \times c$ vote matrix $Z$ as inputs and outputs a $1 \times c$ matrix equal to the sum of the row vectors of $Z$. We call this $1 \times c$ matrix the \textbf{score} for the vote matrix $Z$. The type of the matrix used for the score is also used to describe the score (e.g. a plurality score is a score for a plurality matrix). We call any candidate in $\textbf{X}$ whose corresponding element in a score is equal to the maximum value in the score a \textbf{winner} for that score (there can be multiple winners if multiple candidates are tied for the highest score). Any candidate who is not a winner for that same score is a \textbf{loser} for that score. Any time there is more than one winner for a given score, we call the score a \textbf{tie}.  Winners, losers, scores, and ties can be described by the type of score used to determine them (e.g. a \textbf{plurality winner} is the candidate whose corresponding element in a plurality score is equal to the maximum value in the plurality score).

\flushleft{\textbf{Definition 2.10}}
A \textbf{voting system} (we will often refer to this as simply a \textbf{system} for short) is a method for selecting a single candidate from a set of candidates. In particular, an \textbf{approval system}, \textbf{plurality system}, or \textbf{beta(k) system} is a voting system that randomly selects a candidate from the set of approval winners, plurality winners, or beta(k) winners (respectively).

\section{Beta(k) Properties}

From now on, \textbf{$n_i$} will denote the $i^{th}$ row of a vote matrix, $n$ will denote the number of votes, $c$ will denote the number of candidates, and $\textbf{X} = \{C_1,C_2,\ldots,C_c\}$ will denote the set of candidates. $P$ will denote a $n \times c$ plurality matrix, $A$ will denote a $n \times c$ approval matrix, and $B$ will denote a $n \times c$ beta(k) matrix; $p_j$ will denote the $j^{th}$ element in $E(\textbf{X},P)$, $a_j$ will denote the $j^{th}$ element in $E(\textbf{X},A)$, and $b_j$ will denote the $j^{th}$ element in $E(\textbf{X},B)$.

\begin{lemma}
For any type of score, the set of winners is non-empty.
\end{lemma}
\begin{proof}
By definition, the set of candidates has cardinality $c \geq 1.$ Within the $1 \times c$ score, there must be a maximum element. Therefore, the set of winners is non-empty.
\end{proof}
\begin{theorem}
If
\begin{itemize}
\item $B_{i,j}=k$ if and only if $P_{i,j}=1$, and
\item $k > n$, 
\end{itemize}
then the set of beta(k) winners is a subset of the set of plurality winners.\\
\end{theorem}
\begin{proof} 
Suppose, for the sake of contradiction, there exists a candidate $C_l$ who is a beta(k) winner and not a plurality winner. Let $C_w$ be any candidate in the set of plurality winners. Since $k > n$, we shall say $k = n + \epsilon$ where $\epsilon >0$. \bigskip

Because $C_w$ is a plurality winner and $C_l$ is not, $p_w \geq p_l + 1$. Because $C_l$ is a beta(k) winner, $b_l \geq b_w$ and $b_l \leq p_l \cdot k + (n-p_l)$. \bigskip

Then
$$b_l \leq (p_l)(n + \epsilon) + (n - p_l) < (p_l + 1)(n + \epsilon) \leq p_w(n+\epsilon) \leq b_w.$$

This means $C_l$ is not a beta(k) winner. Therefore, we have arrived at a contradiction.
\end{proof}
\begin{corollary}
If
\begin{itemize}
\item $B_{i,j}=k$ if and only if $P_{i,j}=1$,
\item there is exactly one plurality winner, and
\item $k > n$, 
\end{itemize}
then $C_w$ is a plurality winner if and only if $C_w$ is a beta(k) winner.\\
\end{corollary}
\begin{proof}
If $C_w$ is a beta(k) winner, then $C_w$ must also be the plurality winner as the set of beta(k) winners is a subset of the set of plurality winners. \bigskip

If $C_w$ is the plurality winner, then $C_w$ must also be a beta(k) winner because the set of beta(k) winners is a subset of the set of plurality winners and the set of beta(k) winners is non-empty.
\end{proof}

\begin{remark}
If
\begin{itemize}
\item $B_{i,j}=k$ if and only if $P_{i,j}=1$, and
\item $k > n$, 
\end{itemize}
then the set of plurality winners may not be a subset of the set of the beta(k) winners.
\end{remark}
\begin{proof}
Suppose, for the sake of contradiction, the set of plurality winners is always a subset of the set of the beta(k) winners. The following counterexample disproves this:

\[ B = \begin{bmatrix}
k & 1 \\
0 & k
\end{bmatrix} 
, P = \begin{bmatrix}
1 & 0 \\
0 & 1
\end{bmatrix}
\]

\end{proof}

\begin{lemma}

If
\begin{itemize}
    \item $A_{i,j}=1$ if $P_{i,j}=1$,
    \item $B = K(P,A)$, and
    \item $C_w$ is both a plurality winner and an approval winner,
\end{itemize}
then $C_w$ is a beta(k) winner for any $k \geq 1$. 
\end{lemma}

\begin{proof} Without loss of generality, let $C_l$ be any other candidate in $\textbf{X}$. $p_w \geq p_l$ and $a_w \geq a_l$.\\ 
\bigskip
Thus, A's score under beta(k) is $kp_w + (a_w - p_w)$, while B's score under beta(k) is $kp_l + (a_l - p_l)$.\\
\medskip
A wins beta(k) if $$kp_w + (a_w - p_w) \geq kp_l + (a_l - p_l).$$
This is equivalent to the condition $$k\geq \frac{a_l-a_w+p_w-p_l}{p_w-p_l} = 1 + \frac{a_l-a_w}{p_w-p_l}.$$
Because $a_l - a_w$ is negative and $p_w - p_l$ is positive, $W$ wins for any $k \geq 1$.
\end{proof}

\begin{corollary} Under the same assumptions, suppose $\textbf{X} = \{C_1,C_2\}$, candidate $C_1$ is a plurality winner, and $C_2$ is an approval winner.\\
\bigskip
$C_1$ is the beta(k) winner if $$k > 1 + \frac{a_2-a_1}{p_1-p_2}.$$\\  $C_2$ is the beta(k) winner if $$k < 1 + \frac{a_2-a_1}{p_1-p_2}.$$\\ $C_1$ and $C_2$ are both beta(k) winners if $$k = 1 + \frac{a_2-a_1}{p_1-p_2}.$$ 
\end{corollary}

\begin{corollary}
For any set of candidates $\textbf{X} = \{C_1,C_2,\ldots,C_c\}$, if $C_w$ has a higher score than $C_l$ under both approval and plurality votes, then $C_l$ cannot have the highest score under beta(k) for any $k \geq 1$.\\
\end{corollary}
\begin{proof}
Assume a candidate $C_w$ has a higher score under plurality and approval votes than $C_l$. As shown in the lemma, this implies that $C_w$ has a higher beta(k) score than $C_l$ for all $k>1$, so $C_l$ cannot be a beta(k) winner.\\
\end{proof}

\begin{theorem}
Given a set of candidates $\{C_1, C_2,\ldots,C_c\}$, let \textbf{Y} denote the subset of candidates $\{C_1,C_2,\ldots,C_r\}$ (WLOG) that can win a beta(k) election given a value of $k \geq 1$ where $n$ is the number of voters. Suppose (again, WLOG) $p_1<p_2<\ldots<p_r$. Then $a_r<a_{r-1}<\ldots<a_1$. \bigskip

Conversely, if (again, WLOG) $a_1<a_2<\ldots<a_r$, then $p_r<p_{r-1}<\ldots<p_1$.
\end{theorem}

\begin{proof} Suppose there exists a different ordering of the approval scores. Then for some pair $C_i, C_j (i \neq j)$, $p_i > p_j$ and $a_i > a_j$ or $p_i < p_j$ and $a_i < a_j$. Then, $C_i$ or $C_j$ would not be a potential beta(k) winner, hence not in $\textbf{Y}$.\\
\medskip
To clarify, let $p_1<p_2<\ldots<p_r$. If the approval ordering is not the exact reverse of the plurality ordering i.e. there exists $C_i, C_j$ such that $p_i < p_j$ and $a_i < a_j$. Then $C_j$ is not a potential winner by Corollary 3.3.2. \bigskip

By a similar argument, it is not hard to prove the other direction.
\end{proof}

\begin{lemma}
Suppose as in Theorem 3.4 that \textbf{Y} is the subset $\{C_1,\ldots,C_r\}$ that can win a beta(k) election (i.e. for each of those candidates there exists some $k \geq 1$ such that the given candidate is in the set of winning candidates) and suppose that $r>2$. Suppose that $p_1 < p_2 <\ldots<p_r$ (hence $a_1>\ldots>a_r$). Then the following series of inequalities holds:

$$\frac{a_1-a_2}{p_2-p_1}\leq\frac{a_2-a_3}{p_3-p_2}\leq\ldots\leq\frac{a_{r-1}-a_r}{p_r-p_{r-1}}.$$

\begin{proof}
Suppose that \textbf{Y} is the subset $\{C_1,...,C_r\}$ that can win a beta(k) election and $r>2$. We will proceed to prove this series of inequalities by induction. 
\bigskip

Consider the base case. Since $C_2$ is a potential winner, that implies that there exists a $k$ such that $b_2 \geq b_1$ and $b_2 \geq b_3$. Thus we get the system of inequalities
\bigskip

\[
    \begin{cases}
    k\cdot p_2 + a_2 - p_2 \geq k\cdot p_1 + a_1 - p_1 \\ 
    k\cdot p_2 + a_2 - p_2 \geq k\cdot p_3 + a_3 - p_3, \\
    \end{cases}
\]
\bigskip

which is equivalent to
\bigskip

$$1+\frac{a_1 - a_2}{p_2-p_1}\leq k \leq 1+ \frac{a_2 - a_3}{p_3-p_2}.$$
Thus, we must have
$$\frac{a_1 - a_2}{p_2-p_1} \leq \frac{a_2 - a_3}{p_3-p_2}.$$
Otherwise there would not exist any $k$ for which candidate $C_2$ could win, contradicting the initial assumption.

\bigskip

Now choose $w$ such that $w < r$ and suppose that
$$\frac{a_1-a_2}{p_2-p_1}\leq\ldots\leq\frac{a_{w-1}-a_w}{p_w-p_{w-1}}.$$
Since we assume that $C_w$ is a potential winner, there exists a $k$ such that $b_w>b_{w-1}$ and $b_w>b_{w+1}$. This gives us the system of inequalities
\bigskip

\[
    \begin{cases}
    k\cdot p_w + a_w - p_w \geq k\cdot p_{w-1} + a_{w-1} - p_{w-1} \\ 
    k\cdot p_w + a_w - p_w \geq k\cdot p_{w+1} + a_{w+1} - p_{w+1}, \\
    \end{cases}
\]
\bigskip

which is equivalent to
$$1+\frac{a_{w-1} - a_w}{p_w-p_{w-1}}\leq k \leq 1+ \frac{a_w - a_{w+1}}{p_{w+1}-p_w}.$$
Thus, a $k$ such that $C_w$ is a potential winner exists only if
$$\frac{a_{w-1} - a_w}{p_w-p_{w-1}} \leq \frac{a_w - a_{w+1}}{p_{w+1}-p_w}.$$

This completes the proof by induction.
\end{proof}
\end{lemma}

\begin{theorem}
Suppose as in the above Lemma we have potential candidates $\{C_1,\ldots,C_r\}$ with $r>2$. Then candidate $C_w$ is a beta(k) winner if and only if
$$1+\frac{a_{w-1} - a_w}{p_w-p_{w-1}}\leq k \leq 1+ \frac{a_w - a_{w+1}}{p_{w+1}-p_w}.$$

If $w=1$ or $w=r$, then the left or the right inequality (respectively) should be omitted.
\begin{proof}
($\Rightarrow$)

For candidate $C_w$ to be a beta(k) winner, the equation
$$k\cdot p_w + a_w - p_w \geq
k\cdot p_j + a_j - p_j$$

must be satisfied for all $j \neq w$.

To establish the lower bound for $k$, consider all integers $l$ such that $1\leq l < w$. We must have
$$k\cdot p_w + a_w - p_w \geq
k\cdot p_l + a_l - p_l.$$
Equivalently,
$$k \geq 1 + max\{\frac{a_l - a_w}{p_w - p_l}    |   1\leq l < w \}.$$ 

Notice however, that by the previous Lemma, $\frac{a_l - a_w}{p_w - p_l} \leq \frac{a_{w-1} - a_w}{p_w - p_{w-1}}$ for all such $l$, thus
$$k \geq 1 + \frac{a_{w-1} - a_w}{p_w - p_{w-1}}.$$

To establish the upper bound, we use a similar calculation but for $g$ where $w<g \leq r$. For $C_w$ to win, we need $k$ so that

$$k \leq 1 + min\{\frac{a_w - a_g}{p_g - p_w}    |   w < g \leq r \}.$$

Again, the previous Lemma implies that $\frac{a_g - a_w}{p_w - p_g} \geq \frac{a_{w+1} - a_w}{p_w - p_{w+1}}$ for all such $g$, thus
$$k \leq 1 + \frac{a_{w} - a_{w+1}}{p_{w+1} - p_w}.$$

Therefore, $C_w$ is a beta(k) winner whenever
$$1+\frac{a_{w-1} - a_w}{p_w-p_{w-1}}\leq k \leq 1+ \frac{a_w - a_{w+1}}{p_{w+1}-p_w}.$$
Of course, if $w=1$ then the lower bound for $k$ is just $1$ and if $w=r$ then there is no upper bound for $k$.

($\Leftarrow$)

Suppose, for the sake of contradiction, that $C_w$ is a beta(k) winner and there exist $l, w \in \mathbb{N}$ such that $1 \leq l < w$ and $1 + \frac{a_w - a_l}{p_l - p_w} > k$. \bigskip

If $1 + \frac{a_w - a_l}{p_l - p_w} > k$, then $k \cdot p_w - p_w + a_w  < k \cdot p_l  - p_l + a_l$. \bigskip

However, this contradicts $C_w$ being a beta(k) winner.\\ \bigskip

Similarly, suppose that $C_w$ is a beta(k) winner and there exist $w, g \in \mathbb{N}$ such that $w < g \leq r$ and $k > 1 + \frac{a_g - a_w}{p_w - p_g}$. \bigskip

If $k > 1 + \frac{a_g - a_w}{p_w - p_g}$, then $k \cdot p_w - p_w + a_w  < k \cdot p_g  - p_g + a_g$. \bigskip

Again, this contradicts $C_w$ being a beta(k) winner.

\end{proof}

\end{theorem}

\begin{theorem} 
If
\begin{itemize}
    \item $A_{i,j}=1$ if and only if $B_{i,j} \neq 0$, and
    \item $1 \leq k < 1+ \frac{1}{n}$,
\end{itemize} 

 then the set of beta(k) winners is a subset of the set of approval winners.
\end{theorem}
\begin{proof}
Suppose, for the sake of contradiction, there exists a candidate $C_l$ who is a beta(k) winner and not an approval winner. If $C_w$ is any approval winner, then $a_w \geq a_l + 1$ and $b_w \geq a_w$, but $b_w \leq b_l \leq a_l \cdot (k) < a_l(1+\frac{1}{n}) \leq a_w$. This is a contradiction.
\end{proof}

\begin{corollary}
If
\begin{itemize}
    \item $A_{i,j}=1$ if and only if $B_{i,j} \neq 0$,
    \item there exists exactly one approval winner, and
    \item $1 \leq k < 1+ \frac{1}{n}$,
\end{itemize} 

 then $C_w$ is a beta(k) winner if and only if $C_w$ is an approval winner.
\end{corollary}

\begin{proof}
If $C_w$ is a beta(k) winner, then $C_w$ must also be the approval winner as the set of beta(k) winners is a subset of the set of approval winners. \bigskip

If $C_w$ is the approval winner, then $C_w$ must also be a beta(k) winner because the set of beta(k) winners is a subset of the set of approval winners and the set of beta(k) winners is non-empty.
\end{proof}

\begin{remark}
If
\begin{itemize}
    \item $A_{i,j}=1$ if and only if $B_{i,j} \neq 0$, and
    \item $1 \leq k < 1+ \frac{1}{n}$,
\end{itemize} 

 then the set of approval winners may not be a subset of the set of beta(k) winners.
\end{remark}
\begin{proof}
Suppose, for the sake of contradiction, the set of approval winners is always a subset of the set of the beta(k) winners. The following counterexample disproves this:
\[ B = \begin{bmatrix}
k & 1 \\
k & 1
\end{bmatrix} 
, A = \begin{bmatrix}
1 & 1 \\
1 & 1
\end{bmatrix}
\]
\end{proof}
\section{Voting Criteria}

\subsection{Non-dictatorship}
\textbf{Definition 4.1} A voting system is a \textbf{dictatorship} if there exists a vote whose maximum value always corresponds to a winning candidate under that system.\\
\bigskip

\textit{Observation: beta(k) with at least 3 votes, is a non-dictatorship.}

\begin{proof}
For the sake of contradiction, assume this is a dictatorship. This means there exists $n_j$ such that the maximum value in this vector always corresponds to the beta(k) winner.
WLOG, suppose the maximum value in $n_j$ corresponds to the $w^{th}$ candidate.
Suppose for all $n_p \neq n_j$, $k$ points are assigned to the $l^{th}$ candidate and $0$ points are assigned to the $w^{th}$ candidate. 

Then $b_l > b_w$ and we have arrived at a contradiction.
\end{proof}
\bigskip
\textit{Note: Since approval is equivalent to beta(1), this shows that approval is also a non-dictatorship. Under this construction, multiple winners cannot be selected for a fixed set of plurality votes. Furthermore, since plurality winners agree with beta(k) winners whenever $k>n$, this shows that plurality is also a non-dictatorship (since $n>2$ and all but one of the voters vote for $n_j$ not $n_i$, the plurality winner is unique).}

\subsection{Monotonicity}
\textbf{Definition 4.2} A voting system is \textbf{monotonic} if increasing a value in a vote corresponding to a winning candidate cannot make that candidate lose. Similarly, in a monotonic system, decreasing the value in a vote for a losing candidate cannot make that candidate win. 
\bigskip

\textit{Observation: beta(k) is a monotonic voting system.} 

\begin{proof}
WLOG, assume that beta(k) selects the $w^{th}$ candidate. Then $b_w > b_l$ for all $l \neq w$. Suppose that the $n_j$ vote vector assigns either $0$ or $1$ to its $w^{th}$ element. If this element is increased, then $b_w^{new} > b_w > b_l \geq b_l^{new}$ for all $l \neq w$ where $b_i^{new}$ denotes the new beta(k) total for candidate $i$ after the vote was increased. Similarly, if candidate $l$ does not win, then there exists $w$ such that $b_w > b_l$. Suppose that the $n_j$ vote vector assigns either $1$ or $k$ to its $l^{th}$ element. Then, if this element is decreased, we have $b_l^{new} < b_l< b_w \leq b_w^{new}$, so candidate $l$ still does not win the election.
\end{proof}

\textit{Note: Again, since the result above holds for all k, it holds for approval since beta(1) is equivalent to approval. This also holds for plurality since plurality winners agree with beta(k) winners whenever $k>n$ (assuming that there is not a plurality tie). However, it is trivial to show that even under a tie, the plurality system is still monotonic.}

\subsection{Unanimous winner}
\textbf{Definition 4.3} If there exists an element in $\textbf{X}$ that corresponds to the maximum value in every vote, this element is called a \textbf{unanimous winner}. 
\bigskip

\textit{Observation: beta(k) selects a unanimous winner, if such a candidate exists.} 

\begin{proof}

Assume the $w^{th}$ candidate in the beta(k) votes is the unanimous winner.
Then $b_w \geq b_j$ for all $j \neq w$, and the beta(k) system selects the $w^{th}$ candidate. 

\end{proof}

\textit{Note: Again, since the result above holds for all k, it holds for approval since beta(1) is equivalent to approval and this property holds trivially for plurality elections.}

\subsection{Pareto efficiency}
\textbf{Definition 4.4} For a set of $n$ votes, a system is \textbf{Pareto efficient} if, for any of the system's possible winners $C_w$ and any $l \neq w$, there is some voter who prefers $C_w$ to $C_l$ (i.e. there exists voter $i$ such that $C_w \succ_i C_l$). In English, a Pareto winner is a candidate such that there does not exist a different candidate who is unanimously preferred and a Pareto efficient system is one that always elects a Pareto winner.
\bigskip

\textit{Proposition: beta(k) is pareto efficient when $k > c-1$.} 
\begin{proof}
If a candidate receives at least one k-vote, then at least one voter prefers that candidate to any other candidate. Thus, if we can define bounds for $k$ such that any beta(k) winner is guaranteed to have at least one k-vote, the result follows. We will consider two cases:\bigskip

\textbf{Case 1:} Every candidate receives at least one k-vote. In this case, any winner will be Pareto.\bigskip

\textbf{Case 2:} There is at least one candidate who receives no k-votes, so there are at most $c-1$ candidates that do receive a k-vote. Notice that the candidate(s) who does not receive a k-vote can have no more than $n$ total points under beta(k). Among the other candidates, at least one must have received at least $n/(c-1)$ k-votes. Thus, by setting $k > c-1$, at least one of the candidates with a k-vote will earn more than $n$ points under beta(k). Therefore any winner under such a system must have at least one k-vote.\bigskip

Taking into consideration the cases above, it is clear that for $k > c-1$, any beta(k) winner must be Pareto.

\end{proof}
\textit{Note: It is trivial to see that any plurality winner is Pareto. Such a winner must have at least one plurality vote, implying that at least one voter prefers that candidate to every other candidate. Notice, however, that with the bound established above, it can be shown that approval does not always select a Pareto winner. Consider a case where there is a tie between two candidates under approval and both candidates are approved by the same voters, but one candidate is unanimously preferred over the other. Then both of these candidates are approval winners, but only one is Pareto.}

\section{Conclusion}
In this paper we derived some interesting properties for beta(k) and saw how its performance compares to plurality and approval. In particular, for certain $k$, beta(k) is no "worse" than plurality or approval from the results of 3.2 and 3.7. Furthermore, beta(k) has an advantage over approval by being fully Pareto-efficient for $k > c-1$. This result is not hard to extend to show that if a certain condition is satisfied by either plurality or approval but not the other, then that condition is satisfied under beta(k) for a certain range of $k$. In particular, if the given criterion is satisfied by plurality but not by approval, then the criterion is satisfied under beta(k) for $k$ greater than some lower bound; if the criterion is satisfied by approval but not plurality, then the criterion is satisfied under beta(k) for $k$ less than some upper bound. We conjecture that if a criterion is satisfied by both plurality and approval, then it is satisfied by beta(k) for any $k>1$. Beta(k) could be a good alternative to approval or plurality in elections where the number of voters is large and the number of candidates is small. In such cases, one could choose a $k$ large enough to guarantee Pareto results and this $k$ would still be small enough that the beta(k) would be similar to approval.
\bigskip

There are further ways to determine the viability of beta(k). This would include checking whether it satisfies other common voting criteria (i.e. Condorcet-winner, consistency, et cetera) and seeing how it performs in Monte Carlo simulations. Additionally, it would be interesting to see what possible voting strategies exist for beta(k) and how they affect the outcome of beta(k) elections. A method for determining rational voter strategies may be found in (Myerson and Weber, 1993).

\section{Bibliography}
\hangindent=0.7cm
Ace Project: The Electoral Knowledge Network. (Accessed January 2019) \textit{First Past the Post (FPTP)}. Available at \texttt{http://aceproject.org/main/english/es/esd01.htm}
\bigskip

Bassi, Anna (2015) Voting Systems and Strategic Manipulation: An Experimental Study. \textit{Journal of Theoretical Politics} 27(1): 58-85
\bigskip

Brams, Steven J and Fishburn, Peter (2007) \textit{Approval Voting}. New York: Springer.
\bigskip

Cox, Gary W (1985) Electoral Equilibrium under Approval Voting. \textit{American Journal of Political Science} 29(1): 112-118.
\bigskip

Karlin, Anna R and Peres, (2017) \textit{Game Theory Alive!} Rhode Island: American Mathematical Society.
\bigskip

Myerson, Roger B and Weber, Robert (1993) A Theory of Voting Equilibria. \textit{American Political Science Review} 87(1): 102-114. 
\bigskip

Niou, Emerson (2001) Strategic Voting Under Plurality and Runoff Rules. \textit{Journal of Theoretical Politics} 13(2): 209-227.
\bigskip

Riker, William H (1982) The Two-Party System and Duverger's Law: An Essay on the History of Political Science. \textit{American Political Science Review} 76(4): 753–66. 
\bigskip

\end{document}